\documentclass[DIV=calc,paper=a4,fontsize=11pt,twocolumn]{scrartcl} 

\usepackage[english]{babel}
\usepackage[protrusion=true,expansion=true]{microtype}
\usepackage{amsmath,amsfonts,amsthm}
\usepackage{amssymb,stmaryrd,relsize}
\usepackage[final]{graphicx}
\usepackage{xcolor}
\usepackage[normal,small,hypcap,up,labelfont=bf,textfont=it]{caption}
\usepackage{epstopdf}
\usepackage{subfig}
\usepackage{booktabs}
\usepackage{fix-cm}
\usepackage{amssymb,amsfonts}
\usepackage{dsfont}
\usepackage{bbm}
\usepackage{pstricks}
\usepackage{cite}
\usepackage[utf8]{inputenc}
\usepackage[perpage,symbol*]{footmisc}
\usepackage[varg]{txfonts}
\usepackage{balance}
\usepackage{fancyhdr}
\PassOptionsToPackage{hyphens}{url}\usepackage[pdfencoding=auto,psdextra]{hyperref}
\usepackage{bookmark}
\usepackage{verbatim}
\usepackage{fontenc}
\usepackage{cuted}

\usepackage[shortlabels]{enumitem}

\usepackage{bm}
\usepackage{mathrsfs}

\DeclareCaptionFont{mycolor}{\color[HTML]{000000}}
\theoremstyle{definition}

\newtheorem{thm}{Theorem}

\allowdisplaybreaks  

\newcommand{\real}{{\mathbb R}}
\newcommand{\complex}{{\mathbb C}}

\newcommand{\tbullet}{\mathrel{\raise .2ex\hbox{\tiny$\bullet$}}} 

\newcommand{\hscript}{\mathcal{H}}

\newcommand{\trace}{\mathrm{tr\,}}
\newcommand{\rmave}{\mathrm{Ave\,}}
\newcommand{\rmdiag}{\mathrm{diag\,}}
\newcommand{\hhat}{\widehat{H}}
\newcommand{\ahat}{\widehat{A}}
\newcommand{\uhat}{\widehat{U}}

\newcommand{\ab}[1]{\left|#1\right|}
\newcommand{\brac}[1]{\left\{#1\right\}}
\newcommand{\paren}[1]{\left(#1\right)}
\newcommand{\sqbrac}[1]{\left[#1\right]}
\newcommand{\elbows}[1]{{\left\langle#1\right\rangle}}
\newcommand{\ket}[1]{{|#1\rangle}}
\newcommand{\bra}[1]{{\langle#1|}}

\usepackage{sectsty}													
\allsectionsfont{
\color[HTML]{31ADF3}\usefont{OT1}{phv}{b}{n}
}

\sectionfont{
\color[HTML]{31ADF3}\usefont{OT1}{phv}{b}{n}
}

\usepackage{fancyhdr}												
\usepackage{lettrine}
\newcommand{\initial}[1]{%
\lettrine[lines=3,lhang=0.3,nindent=0em]{
\color[HTML]{31ADF3}
{\textsf{#1}}}{}}

\usepackage{titling}															

\newcommand{\HorRule}{\color[HTML]{31ADF3}
\rule{\linewidth}{1pt}%
}

\pretitle{\vspace{-30pt} \begin{flushleft} \HorRule
\fontsize{34}{34} \usefont{OT1}{phv}{b}{n} \color[HTML]{31ADF3} \selectfont
}
\title{Quantum Conditional Stochastic Processes}					
\posttitle{\par\end{flushleft}\vskip 0.5em}

\preauthor{\begin{flushleft}\large \lineskip 0.5em \usefont{OT1}{phv}{b}{sl} \color[HTML]{31ADF3}}
\author{Stan Gudder\\[8pt]}											
\postauthor{\footnotesize \usefont{OT1}{phv}{m}{sl} \color[HTML]{000000}
Department of Mathematics, University of Denver, Denver, Colorado, USA. E-mail: \href{mailto:sgudder@du.edu}{sgudder@du.edu}\\[10pt]		
\scriptsize\usefont{OT1}{phv}{m}{n} \color[HTML]{31ADF3}{\textbf{Editors: \emph{James F. Glazebrook} \& \emph{Danko D. Georgiev}} }\\[5pt]
\color[HTML]{000000}{Article history: Submitted on March 28, 2026; Accepted on June 22, 2026; Published on June 23, 2026.}
\par\end{flushleft}\HorRule}

\date{}																				

\begin{document}
\maketitle
\thispagestyle{fancy} 			
\initial{Q}\textbf{uantum mechanics contains certain novel mathematical concepts. Among these are complex numbers, Hilbert spaces with their unitary and self-adjoint operators, states represented by complex vectors, superpositions of states, collapse of wave functions, Born's rule for probabilities and others. If we accept that quantum mechanics is probabilistic, then these concepts can be derived and they become secondary. In this work, we begin with what we call a \emph{conditional stochastic process}, which is based on real numbers and probabilities. As we shall see, such processes are defined by three simple axioms. We then use conditional stochastic processes to derive quantum mechanics by employing a correspondence called a \emph{dictionary}. We also show that the converse holds. That is, beginning with a quantum system, we employ the dictionary to derive a conditional stochastic process.\\ Quanta 2026; 15: 13--18.}

\begin{figure}[b!]
\rule{245 pt}{0.5 pt}\\[3pt]
\raisebox{-0.2\height}{\includegraphics[width=5mm]{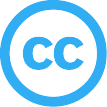}}\raisebox{-0.2\height}{\includegraphics[width=5mm]{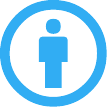}}
\footnotesize{This is an open access article distributed under the terms of the Creative Commons Attribution License \href{http://creativecommons.org/licenses/by/4.0/}{CC-BY-4.0}, which permits unrestricted use, distribution, and reproduction in any medium, provided the original author and source are credited.}
\end{figure}

\section{Introduction}  

Recently there have appeared a considerable number of papers on \emph{conditional stochastic processes} and their relationship to quantum mechanics \cite{bar23,bar24,bar25,bar252,bar26,cal26,fra12,glad20,sksv98}. The most mathematically complete is \cite{cal26}, but we shall follow the interesting papers \cite{bar23,bar25}. However, our work is simpler and more general. It is simpler in the sense that fewer parameters are involved and more general in the sense that it includes an infinite number of configurations and separable infinite dimensional Hilbert spaces. Such Hilbert spaces are general enough to include most of the work on quantum mechanics.

Until 1925, physics involved directly observable quantities such as position, momentum, energy, forces and electromagnetic fields. The main mathematical tool used for solving problems was differential equations. With the advent of quantum mechanics \cite{vN1932,vN1955}, a number of novel mathematical entities seemed to be necessary for physical descriptions. A complex Hilbert space
$\hscript$ appeared, states represented by vectors in $\hscript$, observables given by self-adjoint operators on $\hscript$, dynamics given by unitary operators, complex superpositions of states and state collapse upon measurement, Born's probability rule, etc. Where did all these things come from? The main object of this work is to show that they come from a semi-classical real stochastic process satisfying a few basic axioms. We also show that the converse holds. That is, given a quantum system, we can produce such a stochastic process that generates it.

\section{Conditional Stochastic Processes}  

A \emph{conditional stochastic process} (CSP) is a triple $(K,p,\mu )$ that satisfies the following conditions. We require that
$K=\brac{1,2,\ldots}$ is a countable set of \emph{configurations} and $\mu$~is a nontrivial probability measure on $K$ so that
$\mu (j)\ne 0$ for all $j\in K$. Thus, $0<\mu (j)\le 1$ for all $j\in K$ and \mbox{$\sum\limits _{j=1}^\infty\mu (j)=1$}. We also implicitly assume that this includes the special case in which $K=\brac{1,2,\ldots ,N}$ is finite. The \emph{conditional probability map}
$p\colon K\times\real\times K\to\sqbrac{0,1}\subseteq\real$ is denoted by $p(j,t\mid k)=p(j,t,k)$ where $t$ is time.
We call $p(j,t\mid k)$ the \emph{probability} that the system is at configuration $j\in K$ at time $t$ given it is at $k\in K$ at initial time $0$. We require that $p$ satisfies the \emph{normalization conditions}
\begin{equation}                
\label{eq21}
\sum _{j=1}^\infty p(j,t\mid k)=\sum _{k=1}^\infty p(j,t\mid k)=1
\end{equation}
as well as the \emph{trivialization condition}
\begin{equation}                
\label{eq22}
p(j,0\mid k)=\delta _{jk}
\end{equation}
for all $j,k\in K$, where $\delta _{jk}$ is the Kronecker delta function. The system's \emph{standalone probability distribution}
$p_t\colon K\times\real\to\sqbrac{0,1}$ is
\begin{equation}                
\label{eq23}
p_t(j)=\sum _{k=1}^\infty p(j,t\mid k)\mu (k)
\end{equation}
for every $j\in K$, $t\in\real$. By \eqref{eq21} we have that $p_t$ is a probability distribution because
\begin{align*}
\sum _{j=1}^\infty p_t(j)&=\sum _{j=1}^\infty\sum _{k=1}^\infty p(j,t\mid k)\mu (k) \\
    &=\sum _{k=1}^\infty\mu (k)\sum _{j=1}^\infty p(j,t\mid k)\\
        &=\sum _{k=1}^\infty\mu (k)=1
\end{align*}
Notice, by \eqref{eq22} we have
\begin{equation*}
p_0(k)=\sum _{r=1}^\infty p(k,0\mid r)\mu (r)=\sum _{r=1}^\infty\delta _{kr}\mu (r)=\mu (k)
\end{equation*}
so $\mu (k)$ is the probability the system is initially at configuration $k$. We can write \eqref{eq23} as
\begin{equation*}
p_t(j)=\sum _{k=1}^\infty p(j,t\mid k)p_0(k)
\end{equation*}
for all $j\in K,t\in\real$ which is called the \text{law of total probability}. In summary, a CSP is a triple $(K,p,\mu )$ satisfying
\eqref{eq21} and \eqref{eq22}. It is a simple probabilistic structure based on the real numbers.

A matrix of the form $U_t(j,k)=p(j,t\mid k)^{1/2}$ with indices $j,k\in K$ is the \emph{dynamical map} and is supposed to describe the time evolution of the system. If $f\colon K\to\real$ is a random variable we form the column vector $\ket{f}$ whose $j$th entry is $f(j)$. Then the time evolution of $f$ is given by the vector $\ket{U_tf}$ where $U_tf$ is the matrix product of $U_t$ and $\ket{f}$. But this is very limited since we only have this one dynamical map and there should be various ways a system evolves. This is where the complex numbers enter the situation! If we include complex numbers we have infinitely many possible dynamical maps! Letting $\complex$ be the set of complex numbers and $i$ the imaginary unit, if
$\theta\colon\real\times K^2\to\real$ where $\theta (0,j,k)=0$, we define a \emph{general dynamic map} as a matrix with entries $j,k$:
\begin{equation}                
\label{eq24}
U_t(j,k)=e^{i\theta (t,j,k)}p(j,t\mid k)^{1/2}\in\complex
\end{equation}
Since by \eqref{eq21} we have
\begin{align*}
\sum _{j=1}^\infty\ab{U_t(j,k)}^2 & =\sum _{j=1}^\infty p(j,t\mid k)=\sum _{k=1}^\infty p(j,t\mid k)\\
  &=\sum _{k=1}^\infty\ab{U_t(j,k)}^2=1
\end{align*}
so $U_t$ is a unitary matrix ($U_t^*U_t=U_tU_t^*=I$ with $U_t^*$ the conjugate transpose of $U_t$). This is precisely the type of evolutionary operators in quantum mechanics. We also have by \eqref{eq22} that
\begin{equation*}
U_0(j,k)=e^{i\theta (0,j,k)}p(j,0,k)^{1/2}=e^{i\theta (0,j,k)}\delta _{jk}=\delta _{jk}
\end{equation*}
so the matrix $U_0=I$. We have seen where the complex numbers come from and now consider the Hilbert space.

\section{The Dictionary}  

Let $\hscript =L_2(K,\mu )$ be the separable complex Hilbert space consisting of vectors $\ket{\psi}=(\psi _1,\psi _2,\ldots)$ where
$\psi _j\in\complex$ with squared norm
\begin{equation*}
\|\psi\|^2=\sum _{j=1}^\infty\ab{\psi _j}^2\mu (j)<\infty
\end{equation*}
and inner product $\elbows{\psi,\phi}=\sum\limits _{j=1}^\infty\overline{\psi}_j\phi _j\mu (j)$. We call the orthonormal basis on
$\hscript$ given by 
\begin{equation*}
e_j=\paren{0,0,\ldots ,\tfrac{1}{{\mu (j)}^{1/2}},0,\ldots}
\end{equation*}
where $\tfrac{1}{{\mu (j)}^{1/2}}$ is in the $j$th position, the \emph{configuration basis}. The rank-one projection operators
\begin{equation*}
P_j=\rmdiag (0,0,\ldots ,1,0,\ldots )=\ket{e_j}\bra{e_j}
\end{equation*}
where $1$ is in the $j$th position, are called the \emph{configuration projections}. These operators form a
\emph{projection valued measure} (PVM) and satisfy $P_jP_k=\delta _{jk}P_j$ for all $j,k\in K$ and
$\sum\limits _{j=1}^\infty P_j=I$. Letting $\trace$ be the trace, we have
\begin{align*}
\trace (U_t^*P_jU_tP_k)&=\sum _{r=1}^\infty\elbows{e_r,U_t^*P_jU_tP_ke_r}\\
&=\elbows{e_k,U_t^*P_jU_te_k}\\
&=\elbows{U_te_k,\ket{e_j}\bra{e_j}U_te_k}
=\ab{\elbows{e_j,U_te_k}}^2\\
&=\ab{U_t(j,k)}^2=p(j,t\mid k)
\end{align*}

The \emph{dictionary} \cite{bar23,bar25} given by
\begin{equation}                
\label{eq31}
p(j,t\mid k)=\trace (U_t^*P_jU_tP_k)
\end{equation}
for all $j,k\in K$, $t\in\real$, translates the CSP as expressed by the left-hand side of \eqref{eq31} and the representation as expressed by the right-hand side of \eqref{eq31}.
The standalone probability becomes
\begin{equation}                
\label{eq32}
p_t(j)=\sum _{k=1}^\infty p(j,t\mid k)\mu (k)=\trace\paren{U_t^*P_jU_t\sum _{k=1}^\infty\mu(k)P_k}
\end{equation}
We have thus shown that a CSP describes a complex separable Hilbert space and its concepts are represented by the dictionaries \eqref{eq31}, \eqref{eq32}. We now show that the converse holds. In this way a quantum system describes a CSP. 

Let $\hscript$ be a separable complex Hilbert space. Since two separable Hilbert spaces of the same dimension are isomorphic, we can assume that $\hscript =L_2(K,\mu )$ where $K=\brac{1,2,\ldots}$ and $0<\mu (j)\le 1$ with
$\sum\limits _{j=1}^\infty\mu (j)=1$. Let $U_t$, $t\in\real$ be unitary operators on $\hscript$ with $U_0=I$ and let $P_j$ be the rank-one projections on $H$ defined previously. Define $p\colon K\times\real\times K\to\sqbrac{0,1}$ as in the dictionary \eqref{eq31} by
\begin{equation*}
p(j,t\mid k)=\ab{U_t(j,k)}^2=\trace (U_t^*P_jU_tP_k)
\end{equation*}
and consider the triple $(K,p,\mu )$. Since $U_t$ is unitary we have for all $j\in k$ that
\begin{equation*}
\sum _{k=1}^\infty p(j,t\mid k)=\sum _{k=1}^\infty\ab{U_t(j,k)}^2=1
\end{equation*}
and for all $k\in K$ that
\begin{equation*}
\sum _{j=1}^\infty p(j,t\mid k)=\sum _{j=1}^\infty\ab{U_t(j,k)}^2=1
\end{equation*}
so \eqref{eq21} holds. Also, we have
\begin{equation*}
p(j,0\mid k)=\ab{U_0(j,k)}^2=\delta _{jk}
\end{equation*}
for all $j,k\in K$ so \eqref{eq22} holds. We conclude that $(K,p,\mu )$ is a CSP so the converse holds.

Let $(K,p,\mu )$ be a CSP and let $\hscript =L_2(K,\mu )$ be the corresponding quantum system. We now show that various quantum concepts on $\hscript$ are automatically described by $(K,p,\mu )$. For example, most quantum systems have dynamics given by unitary operators $\uhat _t=e^{it\hhat}$ where $\hhat$ is a self-adjoint operator
on $\hscript$ called the \emph{Hamiltonian}. Assuming that $\hhat$ has pure point spectrum with eigenvalues
$\lambda _j\in\real$ so that
$\hhat e_j=\lambda _je_j$, then $\uhat _te_j=e^{it\lambda _j}e_j$, $=1,2,\ldots\,$. Forming the corresponding CSP we have that $p(j,t\mid k)=\delta _{jk}$ and $U_t(j,k)=e^{it\lambda _j}\delta _{jk}$. In general, if $\ahat$ is a self-adjoint operator on
$\hscript$ with pure point spectrum $\lambda _j$ and $\ahat e_j=\lambda _je_j$, then we can represent $\ahat$ on the corresponding CSP by the random variable $A(j)=\lambda _j$, $j=1,2,\ldots\,$. Then $A$ has \emph{average value} at time $0$ given by
\begin{align*}
\rmave _0 & =\sum _{j=1}^\infty\lambda _jp_0(j)
=\sum _{j=1}^\infty\lambda _j\mu (j)
\end{align*}
which is the usual quantum mechanical value. Conversely, if $A(j)=\lambda _j$, $j=1,2,\ldots$, is a random variable on a CSP, then $\ahat =\sum\limits _{j=1}^\infty\lambda _jP_j$ is a self-adjoint operator on $L_2(K,\mu)$. We refer to our references \cite{bar23,bar24,bar25,bar252,bar26,cal26,fra12}
which show how essentially all the quantum mechanical concepts are reproduced using the dictionary.

\section{Induced States and Tensor Products}  

In this section we show that a CSP induces certain states on its corresponding quantum system. We also show that the natural product of two CSP's generates the tensor product of their corresponding quantum systems.

Let $(K,p,\mu )$ be a CSP and $L_2(K,\mu )$ be the corresponding quantum system. As we have seen these systems are related by the dictionary \eqref{eq31} and the standalone probability \eqref{eq32}. We call
\begin{equation*}
\rho _t=\sum _{k=1}^\infty\mu (k)U_tP_kU_t^*
\end{equation*}
an \emph{induced state} of the system $L_2(K,\mu )$ at time $t$. In this way, different functions $\theta (t,j,k)$ induce various quantum states on $L_2(K,\mu )$. Also, notice that $p_t(j)=\trace (\rho _tP_j)$ for all $t\in\real$, $j\in K$. The following results show that the dictionary can be written in a simpler way and that induced states have a usual form.

\onecolumn
\begin{thm}    
\label{thm41}
{\textrm{(a)}}\enspace For all $t\in\real$, $k\in K$,
\begin{equation*}
U_tP_kU_t^*=\ket{U_te_k}\bra{U_te_k}=P_{U_te_k}
\end{equation*}
is a one-dimensional projection so it is a pure quantum state.\newline
{\textrm{(b)}}\enspace An induced state $\rho _t$ is a trace~1 positive operation (a density operator) so it is a usual quantum state. Also, $\mu (k)$ is an eigenvalue of $\rho _t$ with corresponding eigenvector $U_te_k$.\newline
{\textrm{(c)}}\enspace A dictionary can be written:
\begin{equation*}
p(j,t\mid k)=\ab{\elbows{e_j,U_te_k}}^2
\end{equation*}
and we have
\begin{equation*}
p_t(j)= \sum _{k=1}^\infty\mu (k)\ab{\elbows{e_j,U_te_k}}^2
\end{equation*}
\end{thm}
\begin{proof}  
(a)\enspace For all $t\in\real$, $k\in K$ we have
\begin{equation*}
U_tP_kU_t^*=U_t\ket{e_k}\bra{e_k}U_t^*=\ket{U_te_k}\bra{U_te_k}=P_{U_te_k}
\end{equation*}
which is a one-dimensional projection.\\
(b)\enspace By (a) we have $\rho _t=\sum\limits _{k=1}^\infty\mu (k)P_{U_te_k}$ so $\rho _t$ is a positive operator with
$\trace (\rho _t)=\sum\limits _{k=1}^\infty\mu (k)=1$. Also $\mu (k)$ is an eigenvalue of $\rho _t$ with corresponding eigenvector $U_te_k$ for all $t\in\real$, $k\in K$.\newline
(c)\enspace By (a) and \eqref{eq31} we have
\begin{align*}
p(j,t\mid k)
&=\trace (U_t^*P_jU_tP_k)=\trace (P_jU_tP_kU_t^*) \\
&=\elbows{e_j,U_tP_kU_t^*e_j}
=\elbows{e_j,U_t\mid e_k}\elbows{e_k\mid U_t^*e_j} \\
& =\elbows{e_j,U_te_k}\elbows{U_te_k,e_j}
=\ab{\elbows{e_j,U_te_k}}^2
\end{align*}
Finally \eqref{eq32} gives
\begin{equation*}
p_t(j)=\sum _{k=1}^\infty\mu (k)p(j,t\mid k)=\sum _{k=1}^\infty\mu (k)\ab{\elbows{e_j,U_te_k}}^2
\qedhere
\end{equation*}
\end{proof}

We define 
\begin{equation*}
p_t(\rho _j)=p_t(j)=\trace (\rho _tP_j)=\elbows{e_j,\rho _te_j}
\end{equation*}
This is a special case of \emph{Born's rule}. Extending this by linearity, if $\psi=\sum\limits _{j=1}^\infty c_je_j$ is a unit vector, then
\begin{equation*}
p_t\paren{\ket{\psi}\bra{\psi}}=\trace\paren{\rho _t\ket{\psi}\bra{\psi}}=\elbows{\psi ,\rho _t\psi}
\end{equation*}
A \emph{quantum effect} is a positive operator $E$ on $L_2(K,\mu )$ with $\trace (E)=1$. An effect is a general quantum event. Extending our previous result by linearity, we have that if $E$ is an effect, then $p_t(E)=\trace (\rho _tE)$. This is a general Born's rule for the state $\rho _t$.

We now consider tensor products of Hilbert spaces and show that they occur naturally in this framework. Let
$(K_1,p_1,\mu _1)$ and $(K_2, p_2,\mu _2)$ be CSP's and form the triple
$(K_1\times K_2,p_1\times p_2,\mu _1\times\mu _2)$ where
\begin{align*}
p_1\times p_2(j_1,j_2,t\mid k_1\times k_2)&=p_1(j_1,t\mid k_1)p_2(j_2,t\mid k_2)
\intertext{and}
\mu _1\times\mu _2(j_1\times j_2)&=\mu _1(j_1)\mu _2(j_2)
\end{align*}

\newpage
\begin{thm}    
\label{thm42}
{\textrm{(a)}}\enspace The triple $(K_1\times K_2, p_1\times p_2,\mu _1\times\mu _2)$ is a CSP.\newline
{\textrm{(b)}}\enspace The quantum system for (a) is the tensor product $L_2(K_1,\mu _1)\otimes L_2(K_2,\mu _2)$.
\end{thm}
\begin{proof}  
(a)\enspace For every $t\in\real$, $k_1\in K_1$, $k_2\in K_2$ we have
\begin{equation*}
\sum _{j_1,j_2=1}^\infty p_1\times p_2(j_1\times j_2,t\mid k_1\times k_2)
   =\sum _{j_1=1}^\infty p_1(j_1,t\mid k_1)\sum _{j_2=1}^\infty p_2(j_2,t\mid k_2)=1
\end{equation*}
and similarly for every $t\in\real$, $j_1\in K_1$, $j_2\in K_2$ we have
\begin{equation*}
\sum _{k_1,k_2=1}^\infty p_1\times p_2 (j_1\times j_2,t\mid k_1\times k_2)=1
\end{equation*}
Moreover, we have
\begin{align*}
p_1\times p_2(j_1\times j_2,0\mid k_1\times k_2) =p_1(j_1,0\mid k_1)p_2(j_2,0\mid k_2)
   =\delta _{j_1,k_1}\delta _{j_2,k_2}=\delta _{j_1\times j_2,k_1\times k_2}
\end{align*}
(b)\enspace We show that the quantum system corresponding to the CSP of (a) is
$L_2(K_1,\mu _1)\otimes L_2(K_2,\mu _2)$. Define the unitary operators on $L_2(K_1,\mu _1)$ and
$L_2(K_2,\mu _2)$ given by
\begin{equation*}
U _{rt}(j_r,k_r)=e^{i\theta _r(t,j_r,k_r)}p(j_r,t\mid k_r),\quad r=1,2
\end{equation*}
We then have
\begin{align*}
\trace&(U_{1t}^*\otimes  U_{2t}^*P_{1j}\times P_{2j}U_{1t}\otimes U_{2t}P_{1k}\otimes P_{2k})
   =\trace (U_{1t}^*P_{1j}U_{1t}P_{1k})\trace (U_{2t}^*P_{2j}U_{2t}P_{2k})\\
   &=p_1(j_1,t\mid k_1)p_2(j_2,t\mid k_2)=p_1\times p_2(j_1\times j_2\mid k_1\times k_2)
\end{align*}
Also,
\begin{align*}
\trace&\Biggl(U_{1t}^*\otimes U_{2t}^*P_{1j}\otimes P_{2j}U_{1t}\otimes U_{2t}\Biggr.
   \left.\sum _{k_1,k_2=1}^\infty\mu _1\times\mu _2(k_1\times k_2)P_{1k_1}\otimes P_{2k_2}\right)\\
   &=\trace\paren{U_{1t}^*P_{1j}U_{1t}\sum _{k_1=1}^\infty\mu _1(k_1)P_{1k_1}}
   \trace\paren{U_{2t}^*P_{2j}U_{2t}\sum _{k_2=1}^\infty\mu _2(k_2)P_{2k}}\\
   &=p_1(j_1)p_2(j_2)=p_1\times p_2(j_1\times j_2)\qedhere
\end{align*}
\end{proof}

Finally, let $L_2(K_1,\mu _1)$, $L_2(K_2,\mu _2)$ be two quantum systems and let
$\rho _t^r=\sum\limits _{k_r=1}^\infty\mu _r(k_r)U_t^rP_{k_r}^rU_t^{r*}$, $r=1,2$ be their induced quantum states. Then the induced quantum state for $L_2(K_1,\mu _1)\otimes L_2(K_2,\mu _2)$ is given by
\begin{align*}
\rho _t^1\otimes\rho _t^2
   &=\sum _{k_1=1}^\infty\mu _1(k_1)U_t^1P_{k_1}^1U_t^{1*}\otimes
   \sum _{k_2=1}^\infty\mu _2(k_2)U_t^2P_{k_2}^2U_t^{2*}\\
   &=\sum _{k_1,k_2=1}^\infty\mu _1\times\mu _2(k_1\times k_2)
   (U_t^1\otimes U_t^2)(P_{k_1}^1\otimes P_{k_2}^2)(U_t^1\otimes U_t^2)^*
\end{align*}
which is the induced state corresponding to the CSP $(K_1\times K_2,p_1\times p_2,\mu _1\times\mu _2)$.

\section{Classical Stochastic Processes}  
We have used CSP $(K,p,\mu )$ to derive a quantum system. The purpose of this exercise is that a CSP is defined in terms of real numbers and satisfies three simple axioms \eqref{eq21}, \eqref{eq22}. In contrast, quantum mechanics employs complex numbers and is based on non-intuitive mathematical concepts such as Hilbert spaces, self-adjoint and unitary operators traces and inner products. However, it should be pointed out that a CSP is not really a classical probabilistic concept. This is because a classical stochastic process need not satisfy the condition
\begin{equation}                
\label{eq51}
\sum _{k=1}^\infty p(j,t\mid k)=1
\end{equation}

\twocolumn
To understand why, we define a \emph{classical stochastic process} to be a 4-tuple $(K,p,\mu ,f_t)$, where, as before,
$K=\brac{1,2,\ldots}$ and $\mu$ is a nontrivial probability measure on $K$. The functions $f_t\colon K\to K$, $t\in\real$ are random variables and $f_t(k)$ is the system's configuration at time $t$ when its configuration at initial time $0$ is
$k$. As before, $p\colon K\times\real\times K\to [0,1]$ but we write $p(j,t\mid\mid k)$ to distinguish it from the CSP function $p(j,t\mid k)$. We define $p(j,t\mid\mid k)$ to be the classical conditional probability that the system is at configuration $j$ at time $t$ given it is at configuration $k$ at initial time $0$. This probability is given classically by
\begin{align}   
\label{eq52}
p(j,t\mid\mid k)&=\frac{\mu\sqbrac{\brac{r\in K\colon f_t(r)=j}\cap\brac{r\colon f_0(r)=k}}}
   {\mu\sqbrac{\brac{r\colon f_0(r)=k}}}\notag\\
   &=\frac{\mu\sqbrac{f_t^{-1}(j)\cap f_0^{-1}(k)}}{\mu\sqbrac{f_0^{-1}(k)}} \notag\\
   & =\frac{\mu\sqbrac{f_t^{-1}(j)\cap\brac{k}}}{\mu\sqbrac{\brac{k}}}
\end{align}
where the last equality comes from $f_0^{-1}(k)=k$. The properties of $p(j,t\mid\mid k)$ are given in the next theorem.

\begin{thm}    
\label{thm51}
The classical conditional probability satisfies the following:\newline
{\textrm{(a)}}\enspace $\sum\limits _{j=1}^\infty p(j,t,\mid\mid k)=1$ for all $k\in K$, $t\in\real$,\newline
\smallskip
{\textrm{(b)}}\enspace $p(j,0\mid\mid k)=\delta _{jk}$ for all $j,k\in K$.\newline
\smallskip
{\textrm{(c)}}\enspace $\sum\limits _{k=1}^\infty p(j,t\mid\mid k)=\ab{f_t^{-1}(j)}=\hbox{cardinality of }f_t^{-1}(j)$
for all $j\in K$, $t\in\real$.
\end{thm}
\begin{proof}  
(a)\enspace Applying \eqref{eq52} we have for all $k\in K$, $t\in\real$
\begin{align*}
\sum _{j=1}^\infty p(j,t\mid\mid k) & =\frac{1}{\mu\sqbrac{\brac{k}}}\sum _{j=1}^\infty\mu\sqbrac{f_t^{-1}(j)\cap\brac{k}} \\
   & =\frac{1}{\mu\sqbrac{\brac{k}}}\mu\sqbrac{K\cap\brac{k}}=1
\end{align*}
(b)\enspace By \eqref{eq52} we obtain for all $j,k\in\real$
\begin{equation*}
p(j,0\mid\mid k)=\frac{\mu\sqbrac{f_0^{-1}(j)\cap\brac{k}}}{\mu\sqbrac{\brac{k}}}
   =\frac{\mu\sqbrac{\brac{j}\cap\brac{k}}}{\mu\sqbrac{\brac{k}}}=\delta _{jk}
\end{equation*}
(c)\enspace By \eqref{eq52} we obtain for all $j\in K$, $t\in\real$
\begin{align*}
\sum _{k=1}^\infty p(j,t\mid\mid k)&=\sum _{k=1}^\infty\frac{\mu\sqbrac{f_t^{-1}(j)\cap\brac{k}}}{\mu\sqbrac{\brac{k}}} \\
   & =\sum _{k\in f_t^{-1}(j)}\frac{\mu\sqbrac{\brac{k}}}{\mu\sqbrac{\brac{k}}}\\
   &=\sum _{k\in f_t^{-1}(j)}1=\ab{f_t^{-1}(j)}\qedhere
\end{align*}

We conclude that although $p$ satisfies the usual conditions given by
Theorem \ref{thm51} (a), (b), it satisfies
\begin{equation*}
\sum _{k=1}^\infty p(j,t\mid\mid k)=1
\end{equation*}
if and only if $\ab{f_t^{-1}(j)}=1$ for all $j\in K$ which is equivalent to $f_t$ being a bijection.
\end{proof}

\balance
\interlinepenalty=10000

\end{document}